\documentclass[a4paper,reqno,11pt]{amsart}

\usepackage[a4paper, scale={0.72,0.74}, marginratio={1:1}, footskip=7mm, headsep=10mm]{geometry}

\usepackage[style=alphabetic, backend=bibtex8, citestyle=alphabetic, firstinits=true, maxnames=99, maxalphanames=4, sorting=nyt]{biblatex}

\usepackage{amsmath,amsfonts,amssymb,amsthm,amscd,mathrsfs,mathdots,bbm,color,bm,mathtools}
\usepackage{mathrsfs}
\usepackage{sansmath}
\numberwithin{equation}{section}
\usepackage{enumerate}
\usepackage{comment}
\usepackage{algorithm}
\usepackage{algorithmic}
\usepackage{textgreek}
\usepackage{color,soul}
\usepackage[dvipsnames]{xcolor}
\definecolor{gray1}{gray}{0.9}
\definecolor{gray2}{gray}{0.75}
\definecolor{gray3}{gray}{0.6}
\usepackage{a4wide}
\allowdisplaybreaks

\usepackage{graphicx}
\usepackage[linecolor=white,backgroundcolor=orange,bordercolor=red]{todonotes}

\usepackage[normalem]{ulem}

\usepackage{tikz}
\usetikzlibrary{decorations.pathreplacing}

\usepackage{newtxmath}
\usepackage{newtxtext}
\usepackage{subcaption}

\usepackage{stmaryrd}

\DeclareFontFamily{U}{mathx}{}
\DeclareFontShape{U}{mathx}{m}{n}{<-> mathx10}{}
\DeclareSymbolFont{mathx}{U}{mathx}{m}{n}
\DeclareMathAccent{\widecheck}{0}{mathx}{"71}

\usepackage[colorlinks,citecolor=MidnightBlue,urlcolor=MidnightBlue,linkcolor=MidnightBlue]{hyperref}

\newenvironment{myenumerate}{%
\renewcommand{\theenumi}{(\roman{enumi})}%
\renewcommand{\labelenumi}{\theenumi}%
\begin{list}{\labelenumi}
	{%
	\setlength{\itemsep}{0.4em}%
	\setlength{\topsep}{0.5em}%
	\setlength\leftmargin{2.45em}%
	\setlength\labelwidth{2.05em}%
	\setlength{\labelsep}{0.4em}%
	\usecounter{enumi}%
	}%
	}%
{\end{list}
}
\renewenvironment{enumerate}{
\begin{myenumerate}}%
{\end{myenumerate}}

\def\R{{\mathbb R}}
\def\C{{\mathbb C}}

\def\PR{\mathbb{P}}

\def\EE{{\mathbb E}}

\def\PP{{\mathbb P}}

\newtheorem{thm}{Theorem}

\numberwithin{thm}{section}
\theoremstyle{plain}
\newtheorem{lem}[thm]{Lemma}

\theoremstyle{definition}

\newtheorem{rem}[thm]{Remark}

\newcommand{\be}{\begin{equation}}
\newcommand{\ee}{\end{equation}}

\newcommand{\ben}{\begin{equation*}}
\newcommand{\een}{\end{equation*}}

\def\_#1{\def\next{#1}%
 \ifx\next\risingsign\expandafter\rising\else^{\underline{#1}}\fi}
\def\risingsign{^}
\def\rising#1{^{\overline{#1}}}

\usepackage[most]{tcolorbox}

\usepackage{tikz}
\usepackage{forest}
\forestset{
  default preamble={
   for tree={circle, draw, 
            minimum size=2.1em, % <-- added
            inner sep=1pt} 
  }
}

\usetikzlibrary{positioning}

\title{Maximal Minimal Spacing for Random Points}
\author{}
\usepackage{nomencl}
\makenomenclature

\begin{document}

\author[F.~D.~Cunden]{Fabio Deelan Cunden}
\address{Fabio Deelan Cunden -- Dipartimento di Matematica, Universit\`a degli Studi di Bari, i-70125 Bari, Italy, and INFN, Sezione di Bari, i-70126 Bari, Italy}
\email{fabio.cunden@uniba.it}

\author[N.~Cuppone]{Noemi Cuppone}
\address{Noemi Cuppone -- Department of Mathematics, King’s College London, Strand, London, WC2R 2LS, United Kingdom}
\email{noemi.cuppone@kcl.ac.uk}

\author[G.~Gramegna]{Giovanni Gramegna}
\address{Giovanni Gramegna -- Dipartimento di Fisica, Universit\`a degli Studi di Bari, i-70126 Bari, Italy, and INFN, Sezione di Bari, i-70126 Bari, Italy}
\email{giovanni.gramegna@uniba.it}

\author[P.~Vivo]{Pierpaolo Vivo}
\address{Pierpaolo Vivo -- Department of Mathematics, King’s College London, Strand, London, WC2R 2LS, United Kingdom}
\email{pierpaolo.vivo@kcl.ac.uk}

\maketitle
\begin{abstract}
From $N+1$ random points on a line we wish to select $M+1$ points so as to maximize the minimal spacing between them. We consider an initial configuration with independent and identically distributed  spacings. {Equivalently, the points are arrival times of a generic renewal process.}  For general spacing distributions, and for all $M\leq N$, we derive exact distributional identities for the {maximal minimal} spacing and obtain its asymptotic behavior.
The problem admits a reformulation in terms of a threshold-resetting random walk. The walk advances by successive random increments and is reset to the origin upon exceeding a fixed threshold. The probability that the optimal spacing exceeds a given value coincides with the probability that the walk completes at least $M$ reset cycles within $N$ steps. This yields an exact representation in terms of first-passage functionals of the walk. 
The same mapping suggests a numerical scheme for the max-min spacing problem in the regime of large  $N$ and $M$, whose accuracy is tested against the exact results obtained here.
\end{abstract}

\section{Introduction}
\subsection{The max-min spacing problem {in one dimension}}
Take a collection of $N+1$ points on the  real line, listed in increasing order: $P_0 < P_1 < \cdots < P_N$. 
We would like to choose a subset of $M+1$ points such that the minimal spacing between any two consecutive points is as large as possible.

    If $M=1$, the problem is trivial: we keep only the two extreme points.
    If $M=N$, there is no choice to make: the smallest gap between the original points provides the solution to our problem. 
    The interesting regime lies between those extreme cases, in which we must discard $N-M$ points.
{An important observation is that }any optimal choice must include the extreme points $P_0$ and $P_N$: removing an endpoint can only decrease the total range without decreasing the smallest spacing. 
%\end{itemize}

The problem can be formalised as follows. %Regardless of any randomness, 
Given the initial $N+1$ points, choosing $M+1$ of them is equivalent to choosing an index subset $i:=(i_0,\dots,i_M)$, such that $i_0<i_1<\cdots<i_{M-1}<i_M$. In our case, $i_0:=0$ and $i_M:=N$, as we are always choosing the extremal points. Let 
\begin{equation}\label{I_MN}
   I_{M,N}:=\Bigl\{ i=(i_0,\dots,i_M): 
0=i_0<i_1<\cdots<i_{M-1}<i_M=N \Bigr\} 
\end{equation}
denote the collection of admissible index subsets. There are 
\begin{equation}
    |I_{M,N}|=\binom{N-1}{M-1}
\end{equation}
of such selections.
For $i\in I_{M,N}$, the \textbf{minimal spacing} among the selected points is
\begin{equation}
\label{eq:min_spacing}
    S^{(i)} :=\min_{1\le j\le M}\bigl(P_{i_j}-P_{i_{j-1}}\bigr).
\end{equation}
We define the \textbf{max-min spacing}, as the maximum of the minimal spacing over all possible selections
\begin{equation}
\label{eq:def_maxmin}
S^*:=\max_{i\in I_{M,N}} S^{(i)}.
\end{equation}
The optimal selection may not be unique (see Fig.~\ref{fig:select}).

In order to obtain a {scale invariant measure of dispersion  that can be compared across different samples with different scales and ranges,} we normalize $S^*$ by the total range and thus define the \textbf{relative max-min spacing} as follows
\begin{equation}
\label{eq:def_relmaxmin}
\widetilde S:=\frac{S^*}{P_N-P_0}.
\end{equation}
We are interested in the case of \emph{random} initial configuration of $N+1$ points $\{P_{{n}}\}_{{{n}}=0}^N$. Our goal is to characterize the distribution of the max-min spacing \eqref{eq:def_maxmin}. How does the distribution of $S^*$ depend on $M$ and $N$? What is the typical size of $S^*$? How to characterize typical and atypical fluctuations away from the mean?

\subsection{Motivations}

The max--min spacing $S^*$ sits at the intersection of several classical themes in probability, statistical physics, and combinatorial optimization. 
The classical theory of spacings usually studies the gaps that are already
present in a random configuration~\cite{HajosRenyi1954,Pyke1965,Pyke1972,DavidNagaraja2003,ArnoldBalakrishnanNagaraja1992}. The largest spacing is by now well understood~\cite{Bairamov2010,Deheuvels1982,MijatovicVysotsky2015,Cunden2021}. 

Here the situation is
different: the gaps are not observed but created by selection. We ask how
large the minimal spacing can be made after selecting $M+1$ points from an
initial collection of $N+1$. Equivalently, one seeks the largest gap that
can be enforced simultaneously across $M$ consecutive blocks. The resulting
optimization ranges over strongly correlated configurations{, because different selections are built from overlapping sums of the same underlying random gaps. Nevertheless, it leads to a new
exactly solvable model in the broader setting of extreme-value statistics for
correlated random systems~\cite{MajumdarSchehr2024,SchehrMajumdar2012,MajumdarMounaixSchehr2013}. 

Problems involving unusually large gaps appear throughout probability and
mathematical physics.  In random matrix theory, for example, one studies large
gaps between neighboring eigenvalues or eigenangles~\cite{BenArousBourgade2013,FengWei2025,Dyson1962}. Our setting differs in one
essential respect: the large gaps are produced by a coarse-graining of the
configuration rather than by the original process itself.

There are several equivalent ways to think about the quantity $S^*$.
In statistical physics, $S^*$ is the largest hard-core exclusion radius
compatible with keeping $M+1$ particles from a configuration of $N+1$ {particles},
connecting the problem to R\'enyi's classical parking model~\cite{Renyi1958}.
In ecology, if $N+1$ individuals compete for territory and exactly $M+1$
survive, $S^*$ is the largest minimum territory each survivor can claim.
In operations research, the problem is a one-dimensional instance of the
$p$-dispersion problem~\cite{Kuby1987,ErkutNeuman1989,RaviRossSchneider1994}---select
$M+1$ facilities from $N+1$ candidate sites to maximize the minimum pairwise
distance---which is NP-hard in general metric spaces{~\cite{Erkut1990,SayahIrnich2017,FrancoUchoa2015}} but exactly solvable on
the line{~\cite{RaviRossSchneider1994,AkagiEtAl2018,araki2022max}}; our distributional results provide analytical benchmarks for random
instances, complementing the heuristic literature reviewed in~\cite{MartiEtal2022}.

The present work is also related to the companion paper~\cite{CCGV2026}, where
the objective is not the minimal spacing but the total dispersion (sum of distances). Both belong
to the broad family of maximum diversity/dispersion problems
\cite{FernandezEtal2013,MazzarasiEtal2021,MartiEtal2022}, although the probabilistic structure
turns out to be rather different.

\subsection{Outline of the paper}

Section~\ref{sec:model} introduces the model and its reformulation in terms of a threshold-resetting random walk. The main results and their asymptotic consequences are stated in Section~\ref{sec:results}. Section~\ref{sec:exactlysolvable} is devoted to the analysis of explicitly solvable models. Numerical aspects are considered in Section~\ref{sec:num}, where we describe an approximate algorithm for large instances.  Section~\ref{sec:proofs} contains the proofs. Finally, in section \ref{sec:conclusions} we offer some conclusions and outlook for future research.

\begin{figure}
    \centering
\includegraphics[width=0.7\linewidth]{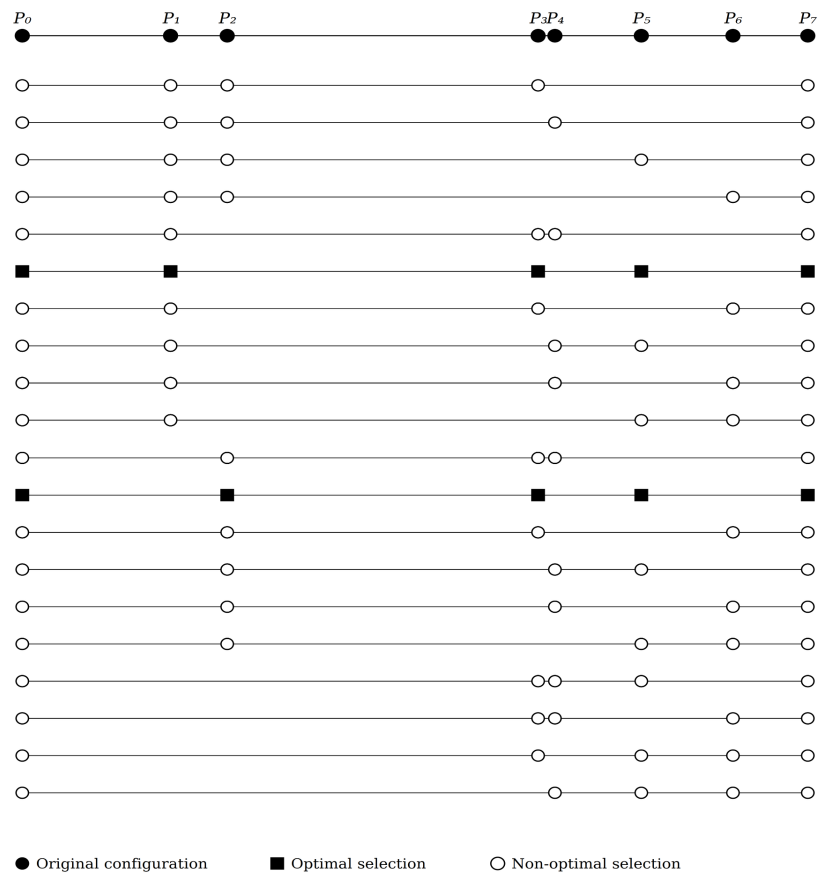}
    \caption{Choose $M+1$ points out of $N+1$ so as to maximize their minimal spacing. Here $N=7$, and $M=4$. The top line contains the initial configuration (black circles). The $\binom{6}{3}=20$ possible selections of points are in the lines below. The optimal ones (the max-min spacing  in this case is $(P_5-P_3)$) are in black squares, while the others are in white circles.}
    \label{fig:select}
\end{figure}
\section{Definition of the model}\label{sec:model} 

The problem is translation invariant; there is no loss of generality in fixing the leftmost point  $P_0=0$. We consider a random configuration of points on the line whose successive spacings are independent and identically distributed (i.i.d.).

Let $(T_j)_{j\geq1}$ be i.i.d. positive  random variables. Let, 
\begin{equation}
\label{def:RWP_i}
P_0:=0,\quad P_1:=T_1, \quad P_2:=T_1+T_2,\quad \ldots \quad P_N:=T_1+\cdots+T_N,\quad \ldots
\end{equation}

Thus $(P_n)_{n\geq0}$ is a \textbf{discrete-time random walk} with positive i.i.d. increments $T_j$'s. The points $(P_n)_{0\leq n\leq N}$ form our initial random configuration.

In this setting, for each choice of indices $i\in {I_{M,N}}$, the minimal spacing defined in~\eqref{eq:min_spacing} is  
\begin{equation}
\label{eq:min_spac_gaps}
    S^{(i)}=\min_{1\leq j\leq M}(T_{{i_{j-1}+1}}+\cdots+T_{i_j}).
\end{equation}
The max-min spacing 
\begin{equation} \label{def:MaxMinModel}
    S^*=\max_{i\in I_{M,N}} \min_{1\leq j\leq M}(T_{{i_{j-1}+1}}+\cdots+T_{i_j}) 
\end{equation}
is therefore the maximum of $|I_{M,N}|$ random variables correlated through the overlap induced by the selection of indices. Our main result shows that the seemingly complicated optimization problem admits a remarkably simple generating function representation.

\subsection{Threshold-resetting framework}

The max-min spacing problem can be mapped exactly to a \emph{threshold-resetting problem} \cite{BiswasMajumdarPal2025,BiroliMajumdarSchehr2026}: a discrete-time random walk on the positive half-line, which starts from zero and resets to zero every time it crosses a threshold $s>0$. This mapping allows for a representation of the tail distribution $\PP(S^*\geq s)$ of the max-min spacing in terms of the generating function of the first passage time of this auxiliary random walk.

The
\textbf{first-passage time} of $(P_n)_{n\geq0}$ at level $s>0$ is 
\begin{equation}\label{def:FirstPassageTime}
\tau_1{(s)}:=\min_{n> 0}\left\{n  \colon  P_n \geq s\right\} .
\end{equation}
In order to define the associated {reset-to-zero process}, we define the subsequent \textbf{crossing times} through the threshold $s>0$
\begin{equation}\label{def:ResetTime}
    \tau_{k}{(s)}:=\min_{  n   > \tau_{k-1}(s)}\left\{n  \colon   P_n-P_{\tau_{k-1}(s)}\geq s\right\},\quad k=1,2,\ldots,  
\end{equation} 
with $\tau_0(s):=0$. 
So, $\tau_1(s)$ is the first-passage time of the random walk at level $s$, $\tau_2(s)$ is the first time that, starting from $P_{\tau_1(s)}$, the random walk has another excursion larger or equal than $s$, and so on. The sequence of stopping times $(\tau_k(s))_{k\geq0}$ forms a renewal process.
Now we define the \textbf{reset-to-zero process} $\left(X_{{n}}^s\right)_{{{n}}\geq0}$ as  
\begin{equation}\label{def:ResetToZeroDef1}
    X_{{n}}^s:=P_{{n}}-P_{\Gamma_{{n}}^s} ,\quad \text{where $\Gamma_{{n}}^s:=\max_{k\geq0}\{\tau_k(s)  \colon  \tau_k(s) \leq {{n}}\}$}.
\end{equation}
Therefore, $X_{{n}}^s$ tracks the excursion since the most recent reset to the origin, and it is set to zero again as soon as it crosses $s$ (see Figure \ref{fig:RWResetting}). Alternatively, we can write
\begin{equation}\label{def:ResetToZeroDef2}
    X_0^s:=0 ,\quad X_{{n}}^s:=
    \begin{cases}
        X_{{{n}}-1}^s+T_{{n}}  &\text{if $X_{{{n}}-1}^s+T_{{n}}< s$}\\
        0  &\text{if $X_{{{n}}-1}^s+T_{{n}}\geq  s$}
    \end{cases} ,\quad {{n}}\geq 1.
\end{equation}
The reset-to-zero process $
\left(X_{{{n}}}^s\right)_{{{n}}\geq0}$ is a Markov chain on $[0,s]$.

Whenever the random walk $X_{{n}}^s$ resets to the origin, we say that it completed a {cycle}. Each cycle has {length}
\begin{equation}
    L_k(s):=\tau_k(s)-\tau_{k-1}(s), \quad \text{for $k\geq 1$}.
\end{equation}
Cycle lengths $\{L_k(s)\}_{k\geq 1}$ are i.i.d. discrete positive random variables and each cycle is independent of the others. Therefore, $L_k(s)$ is the {first passage time} through the threshold $s$ of the $k$-th independent random walk starting from the origin. See Figure~\ref{fig:RWResetting}.

We define  the \textbf{number of complete cycles} of the reset-to-zero process  $\left(X_{{{n}}}^s\right)_{{{n}}\geq0}$ up to time $N$
\begin{equation}  
\label{defKN}     K_N(s)=\max_{k\geq0}\left\{k  \colon \tau_k(s)\leq N\right\}.
\end{equation}
The key identity is: the max-min spacing \eqref{def:MaxMinModel} is at least $s$, if and only if the number of complete excursions of the reset-to-zero process \eqref{defKN} with threshold $s$ up to time $N$ is at least $M$. 
\begin{lem}
\label{lem:equiv}
   $ S^*\geq s \iff K_N(s)\geq M \iff \tau_M(s)\leq N$ .
\end{lem}
Note that the above Lemma is an equality between events $\{S^*\geq s\}=\{K_N(s)\geq M\}=\{\tau_M(s)\leq N\}$ (the probability distribution of the $T_j$'s plays no role).

\begin{figure}
    \centering
\includegraphics[width=0.85\linewidth]{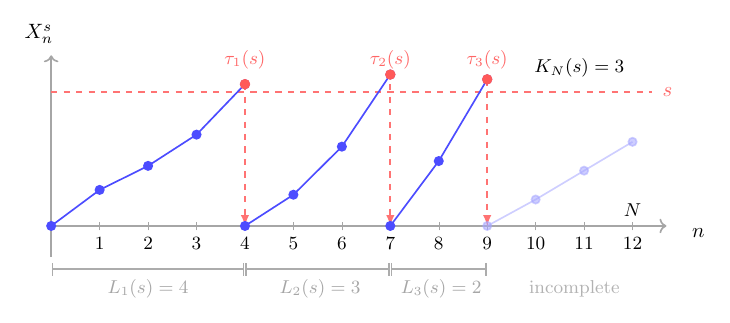}
\caption{\label{fig:RWResetting} Sketch of the threshold-resetting random walk. In this example, there are $K_N(s)=3$ complete cycles of length $L_1(s)=4$, $L_2(s)=3$, $L_3(s)=2$ in $N=12$ steps.}
\end{figure}

\section{Main results}
\label{sec:results}

\subsection{Distribution-free exact formula}

Our main result is an exact formula for the tail distribution of $S^*$, valid for all distributions of the i.i.d. gaps $T_j$'s.
 A central role is played by first-passage time at level $s>0$, that we rewrite as  
%For $s> 0$, define the  \textbf{} as 
\begin{equation}
    \tau_1(s)=\min\left\{k\geq 1\colon T_{1}+\cdots+T_{k}\geq s\right\}.
\end{equation}
Its probability generating function (PGF) is
\begin{equation}\label{eq:pgf_tau^s}
  p_{s}(z) = \EE[ z^{\tau_1(s)} ] = \sum_{k\geq 1} \PP\left(\tau_1(s)=k\right)z^k.
\end{equation} 
For a power series $a(z)=a_0+a_1z+a_2z^2+\cdots$, we write $[z^n]a(z):=a_n$  for the coefficient of $z^n$.

\begin{thm}[Distribution-free formula for the max-min spacing]
\label{thm:main}
For all $1\leq M\leq N$:
\begin{equation}\label{eq:main}
\PR(S^* \geq s)
   = [z^N]\frac{p_s(z)^M}{1-z}.
%    = [z^N](1-z)^{-1}p_s(z)^M.
\end{equation}

\end{thm}
For a numerical illustration, see Fig.~\ref{fig:max-min}.

\begin{lem}[Telescopic formula] 
\label{lem:telescopic} The distribution of the first-passage time $\tau_1(s)$ is 
\begin{equation}
  \PR(\tau_1(s)=k)= \PR(P_{k-1}< s)-\PR(P_{k}< s).
  \label{eq:PL}
\end{equation}
In particular,
\begin{equation}
\begin{aligned}
       p_s(z)
        &=z-(1-z)\sum_{k\geq1}\PR(P_{k}< s)z^{k}.
    \label{eq:pgf-telescopic} 
\end{aligned}
\end{equation}
 Moreover, if $\PP(T_1=0)=0$, then $p_s(z)$ is analytic in the whole complex plane $z\in\C$.
\end{lem}

Let us spell out the {procedure to compute the tail distribution~\eqref{eq:main} of $S^*$}. The model is defined by the distribution of the i.i.d. {spacings} $T_j$'s. The distribution of the partial sums
$P_k=T_1+\cdots+T_k$ is the $k$-fold convolution of the distribution of the $T_j$'s. Evaluating
the distribution functions of these partial sums at the threshold $s$ gives,
through the identity \eqref{eq:PL}, the law of the first-passage
time $\tau_1(s)$. Summing these probabilities with weights $z^k$ gives the
generating function $p_s(z)$ in \eqref{eq:pgf-telescopic}. Finally, inserting
$p_s(z)$ into the coefficient formula \eqref{eq:main} yields the tail
probability $\PP(S^*\geq s)$ of the max--min spacing. Schematically: 
\[
    \PP(T\geq t)
\,\stackrel{\text{convolution}}{\rightsquigarrow}\,
    \PP(P_k\geq s)=\PP(T_1+\cdots+T_k\geq s)
\,\stackrel{\eqref{eq:PL}}{\rightsquigarrow}\,
    \mathcal \PP(\tau_1(s)=k)
\,\stackrel{\eqref{eq:pgf-telescopic}}{\rightsquigarrow}\,
    p_s(z)
    \,	\stackrel{\eqref{eq:main}}{\rightsquigarrow}\,
    \PP(S^*\geq s).
\]

\subsection{Saddle-point asymptotics}
Formula~\eqref{eq:main} is amenable to a saddle-point calculation in the limit of large $N$ and $M$ with fixed ratio $M/N=\alpha$. Since $p_s(z)^M/(1-z)$ is analytic in the unit disk $\{z\in\C\colon |z|< 1\}$, Cauchy's integral formula gives
\begin{equation}\label{eq:P}
\PR(S^*\ge s)
  = \frac{1}{2\pi \mathrm{i}}\oint_{\gamma}
  \frac{p_s(z)^M}{(1-z)\,z^{N+1}}\,dz= \frac{1}{2\pi \mathrm{i}}\oint_{\gamma}
  \frac{e^{Ng(z)}}{(1{-}z)\,z}\,dz,
\end{equation}
where $\gamma$ is any circle centred at $z=0$ and contained in the unit disk, and $g(z)=\alpha\log p_s(z)-\log z$.

For $z>0$, define the  tilted first-passage time $\tau_1(s,z)$ by
\begin{equation}
   \mathbb{P}(\tau_1(s,z) = k)
= \frac{\mathbb{P}(\tau_1(s) = k)z^k }{p_s(z)}. 
\end{equation}
Its generating function is 
\begin{equation}
  \EE[w^{\tau_1(s,z)}] =\sum_{k\geq 1} \PP\left(\tau_1(s,z)=k\right)w^k= p_s(wz)/p_s(z).
\end{equation}
When $z=1$, one recovers the original law: $\tau_1(s,1)=\tau_1(s)$. Values $z>1$ bias the distribution toward larger passage times, while $z<1$ favor smaller ones.

The mean and the variance are:
\begin{equation}\label{eq:TiltedMeanVar}
\mathbb{E}[\tau_1(s,z)]
= \frac{z\;p_s'(z)}{p_s(z)},
\qquad
\operatorname{Var}[\tau_1(s,z)]
 = z\frac{d}{dz}\mathbb{E}[\tau_1(s,z)].  
\end{equation}
The saddle point equation (see Sec.~\ref{sec:proofs}) becomes
\begin{equation}\label{eq:saddle}
\alpha \mathbb{E}[\tau_1(s,z)]=1.
\end{equation}
Define the `typical value'  $s^*$ of $S^*$ by the equation 
\begin{equation}\label{eq:typical}
\alpha \mathbb{E}[\tau_1({s^*},1)]=1.
\end{equation}
We can now state the large deviation theorem.
\begin{thm}[Large deviations]
\label{thm:LD}
Let $(T_j)_{j\geq1}$ be i.i.d. strictly positive random variables, and 
let $\alpha\in (0,1)$. Assume that, for each value of $s$, 
the saddle-point equation
\begin{equation}
\label{eq:assump1}
    \alpha \EE[\tau_1(s,z)]=1
\end{equation}
admits a positive solution $z(s)>0$, and that
\begin{equation}
\label{eq:assump2}
    \operatorname{Var}(\tau_1(s,z(s)))>0.
\end{equation}
Let $s^*$ be defined by~\eqref{eq:typical}. Suppose that $N\to\infty$ with $M=\lfloor \alpha N\rfloor$. Then, for $s>s^*$,
\begin{equation}\label{eq:right-box}
 \PP(S^*\ge s) =
  \frac{e^{-N\psi(s)}}{(1{-}z(s))\,
  \sqrt{2\pi \alpha N\sigma^2(s)}}\left[1+o(1)\right],
\end{equation}
while for $s<s^*$,
\begin{equation}\label{eq:left-box}
\PP(S^*<s) =
  \frac{e^{-N\psi(s)}}{(z(s){-}1)\,
  \sqrt{2\pi \alpha N\sigma^2(s)}}\left[1+o(1)\right].
\end{equation}
Here, 
    \begin{equation}\label{eq:psi}
    \psi(s)=\ln z(s)-\alpha\log p_s(z(s)),\qquad\sigma^2(s) =\operatorname{Var}[\tau_1(s,z(s))]. 
    \end{equation}
\end{thm}
 Equation~\eqref{eq:saddle} selects the tilted law under which the typical value of $\tau_1(s,z)$ matches the constraint $M/N=\alpha$. The theorem identifies not only the  rate function $\psi(s)$ governing the probability of atypical fluctuations of $S^*$ away from its typical value \cite{Touchette2009}, but also the subleading prefactors arising from Gaussian fluctuations around the saddle point~\cite{Daniels1954}.  
For an illustration of the large deviation law above in the exponential case, see Fig. \ref{fig:LDF}.

\begin{rem}
      Assumptions~\eqref{eq:assump1} and~\eqref{eq:assump2} are satisfied whenever the law of the increments $T_j$'s has strictly positive density on $(0,+\infty)$.  This covers a broad and natural class of models.
\end{rem}

\begin{figure}
    \centering
\includegraphics[width=0.75\linewidth]{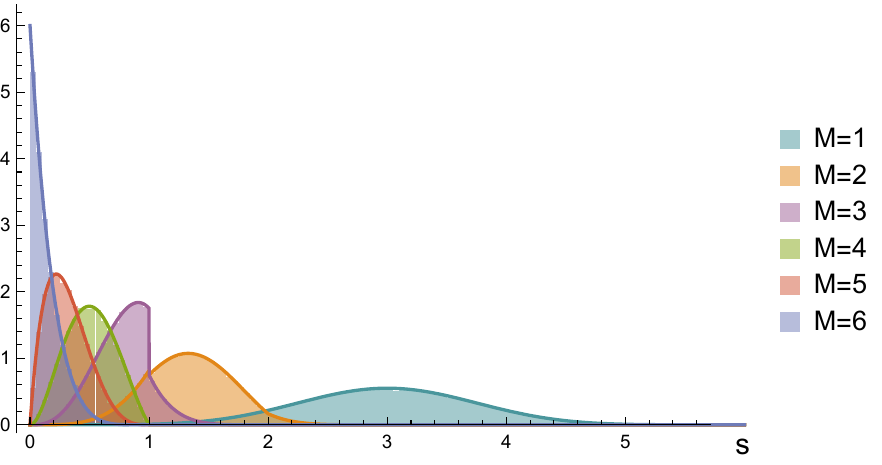}
    \caption{Max-min spacing $S^*$ for gaps $(T_j)_{j\geq1}$ uniformly distributed in the interval $[0,1]$. Numerical simulations (histograms) compared to the explicit formula~\eqref{eq:main}. Here $N=6$, and the sample size is $10^5$. {Note that for uniform random variables in $[0,1]$, the distribution $\PP(P_k<s)$  is defined piecewise, and has jumps in its derivatives at integer values of $s$. This is clearly visible for $M=3$, where the density of $S^*$ has a jump at $s=1$.} }
    \label{fig:max-min}
\end{figure}

\section{Exactly solvable cases}\label{sec:exactlysolvable}

The two models below can be analyzed explicitly because the first-passage
generating function $p_s(z)$ can be computed in closed form.
In both cases this
ultimately comes from a memoryless structure of the increments. Exponential gaps lead to Poisson counting on the line, while geometric gaps produce the corresponding discrete
binomial picture (see Remark~\ref{rem:memoryless}).

\subsection{Notation}
We use the symbol $\stackrel{d}{=}$ to denote \textbf{identity in distribution}, and $\stackrel{d}{\to}$ for \textbf{convergence in distribution}. 
A \textbf{Gamma} random variable of shape $n>0$ and rate $\lambda>0$, denoted $\operatorname{Gamma}(n,\lambda)$, is a random variable with  density $f(x)=\frac{\lambda}{\Gamma(n)}(\lambda x)^{n-1}e^{-\lambda x}$, for $x>0$. When $n=1$, this reduces to an \textbf{exponential} random variable with rate $\lambda$, denoted $\operatorname{Exp}(\lambda)$.
 A \textbf{Beta} random variable with parameters $a,b>0$, denoted $\operatorname{Beta}(a,b)$, is a random variable with  density $f(x)=\frac{\Gamma(a+b)}{\Gamma(a)\Gamma(b)}x^{a-1}(1-x)^{b-1}$, for $0<x<1$.
Finally, we denote by $\mathcal{N}(\mu,\sigma^2)$ a \textbf{Gaussian} random variable with mean $\mu \in \R$ and variance $\sigma^2>0$.

\subsection{Exponential gaps}
\label{sec:exp} 
If $T_i$ is exponential with rate $1$, then each block sum
$
T_{i_{j-1}+1}+\cdots+T_{i_j}
$
is Gamma distributed with shape $n_j=i_j-i_{j-1}$ ($j=1,\ldots,M$) and unit rate, and the $M$ block sums are independent. 
Hence $S^{(i)}$ in Eq. \eqref{eq:min_spac_gaps} is the minimum of $M$ independent Gamma variables with shapes $n_j$ and rate $1$. The survival function is
\begin{equation}
    \mathbb P(S^{(i)}\ge s)
=
\prod_{j=1}^M \mathbb P(T_{i_{j-1}+1}+\cdots+T_{i_j}\ge s)
=
\prod_{j=1}^M \overline{F}_{\Gamma(n_j,1)}(s),
\end{equation}
{where $\overline{F}_{\Gamma(n_j,1)}(s)$ is the tail cumulative distribution of the $T$-sum.}
The corresponding density is obtained by differentiation
\begin{equation}
\label{eq:density_S_i}
f_{S^{(i)}}(s)
=
\sum_{j=1}^M
f_{\Gamma(n_j,1)}(s)
\prod_{\ell\ne j}\overline F_{\Gamma(n_\ell,1)}(s),\qquad s>0,
\end{equation}
where
\[
f_{\Gamma(n,1)}(s)=\frac{s^{n-1}e^{-s}}{(n-1)!},
\qquad
\overline F_{\Gamma(n,1)}(s)=e^{-s}\sum_{r=0}^{n-1}\frac{s^r}{r!} .
\]

The distribution of $S^*$ turns out to be considerably simpler than the individual densities~\eqref{eq:density_S_i}.

\begin{thm}[Max-min spacing for exponential gaps]
\label{thm:main-exp}
    Let $0=P_0<P_1<\cdots< P_N$ as above with gaps $(T_j)_{j\geq1}$ i.i.d. exponential with  rate $1$. Then, for any integer $M=1,\ldots,N$,
    \begin{equation}
    \label{eq:identity_S^*}
        S^*\stackrel{d}{=}\operatorname{{Gamma}}\left(N-M+1,M\right) ,
    \end{equation}
    and for all $M=2,\ldots,N$,
     \begin{equation}
      \label{eq:identity_Stilde}
    M \widetilde{S}\stackrel{d}{=}\operatorname{Beta}(N-M + 1,M-1) .
\end{equation}
    Equivalently, $S^*$ has density
    \begin{equation}
    \label{eq:density_S^*}
   f_{S^*} (s)= \frac{M}{(N-M)!}(Ms)^{N-M}e^{-M s},\quad s\geq0 ,
\end{equation}
while the relative max-min spacing $\widetilde S$  has density
\begin{equation}
\label{eq:density_Stilde}
       f_{\widetilde{S}} (s)  
    = \binom{N-1}{M-1} M(M-1) (Ms)^{N-M} \left(1- M s\right)^{M-2},\quad 0\leq s\leq \frac{1}{M}.
\end{equation}
\end{thm}
For a numerical illustration, see Fig.~\ref{fig:relmaxmin}.
The mean and variance are:
\begin{align}
    E S^*&=\frac{N-M+1}{M},%\sim\frac{1-\alpha}{\alpha }&
    & \operatorname{Var} S^*&=\frac{ N-M+1}{M^2},\\%\sim\frac{1-\alpha}{\alpha^2}\frac{1}{N},\\
    E \widetilde{S}&=\frac{N-M+1}{MN},%\sim\frac{1-\alpha}{\alpha }\frac{1}{N}&
   &  \operatorname{Var} \widetilde{S}&=\frac{(M-1) (N-M+1)}{M^2 N^2(N+1)}.%\sim\frac{1-\alpha}{\alpha}\frac{1}{N^3}.
\end{align}

Note that if $(T_j)_{j\ge1}$ are i.i.d.\ $\operatorname{Exp}(1)$ variables, then  
$P_j=T_1+\cdots+T_j$, $j\geq1$  are the arrival times of a 
unit–rate Poisson process on $\mathbb R_{\ge0}$. The identity in distribution~\eqref{eq:identity_S^*} is equivalent to the equality
\begin{equation}\label{eq:claimed_identity}
   \PP(S^*\geq s)
   =\sum_{k=0}^{N-M}\frac{(Ms)^k}{k!}e^{-Ms},
\end{equation}
or, equivalently,
\begin{equation}
    \label{eq:claimed_identity_Poisson}
    \PP(S^*\geq s)
=
\PP\!\left(\mathrm{Poisson}(Ms)\le N-M\right),
\end{equation}
where $\operatorname{Poisson}(\lambda)$ denotes a \textbf{Poisson} random variable with parameter $\lambda>0$, with distribution $p(k)=\frac{\lambda^k}{k!}e^{-\lambda}$, for $k=0,1,\ldots$. 
\begin{figure}
    \centering
\includegraphics[width=0.75\linewidth]{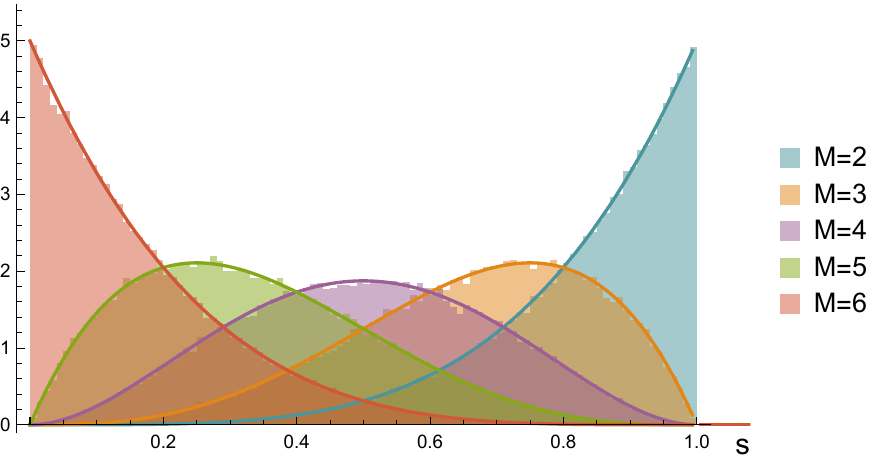}
    \caption{Relative max-min spacing $M\widetilde S$ for exponential gaps. Numerical simulations (histograms) compared to the explicit formula~\eqref{eq:identity_Stilde}. Here $N=6$, and the sample size is $10^5$.}
    \label{fig:relmaxmin}
\end{figure}
\begin{rem}
\label{rmk:unif}
    The relative max-min spacing $\widetilde{S}$ for the model above coincides with the max-min spacing for $N-1$ uniform  
    points on the unit interval.
More precisely, let $U_1,\ldots, U_{N-1}$ be i.i.d. random variables uniformly distributed in the interval $[0,1]$, and denote  the corresponding order statistics by $0=:U_{0:N-1}<U_{1:N-1}<U_{2:N-1}<\cdots<U_{N-1:N-1}<U_{N:N-1}:=1$. Set $P_i:=U_{i:N-1}$, for  $i=0,\ldots,N$. Then, for all $2\leq M\leq N$, the max-min spacing between these points has density~\eqref{eq:density_Stilde}.
\end{rem}
\begin{rem}
The cases $M=1$ and $M=N$ 
 of Theorem~\ref{thm:main-exp} 
   reduce to classical identities: 
\begin{enumerate}
    \item When $M=1$, the only possible selection of indices is $i=(i_0,i_M)=(0,N)$, and the max-min spacing $S^*$ is simply the range of the sample, $P_N-P_0$. This is the sum $T_1+\cdots+T_N$ of $N$ i.i.d. exponential random variables with unit rate, that is a variable $\operatorname{Gamma}(N,1)$.
    \item For $M=N$, the only selection of indices is $i=(0,1,\ldots,N)$, and the max-min spacing $S^*$ is nothing but the minimal spacing in the original sample. This is the minimum $\min\{T_1,\ldots,T_N\}$ of $N$ i.i.d. exponential random variables with unit rate, namely   $\operatorname{Gamma}(1,N)$. 
\end{enumerate}  
\end{rem}

\begin{rem}
It is natural to compare the distribution of $S^*$ with that of the minimal spacing $S_\mathrm{rand}$  
obtained when the $M+1$ selected points are chosen at random. Let $S_{\rm rand}$ denote the minimal spacing obtained by choosing the
index set $i\in I_{M,N}$ uniformly at random, independently of the gaps.
For exponential gaps, $S_{\rm rand}$ is not in general exponentially
distributed. Conditional on the block lengths
$
    n_j=i_j-i_{j-1}$ ($j=1,\ldots,M$),
the $M$ retained spacings are independent Gamma random variables with
respective shapes $n_1,\ldots,n_M$ and unit rate. Hence
\begin{equation}
    \PP(S_{\rm rand}\geq s\mid n_1,\ldots,n_M)
    =
    e^{-Ms}\prod_{j=1}^M
    \sum_{r=0}^{n_j-1}\frac{s^r}{r!} .
\end{equation}
Since choosing $i$ uniformly from $I_{M,N}$ is equivalent to choosing uniformly
a composition $n_1+\cdots+n_M=N$ into $M$ positive parts, we obtain
\begin{equation}
\label{eq:Srand-tail-compositions}
    \PP(S_{\rm rand}\geq s)
    =
    \frac{e^{-Ms}}{\binom{N-1}{M-1}} \sum_{\substack{n_1,\ldots,n_M\geq1\\ n_1+\cdots+n_M=N}}
    \prod_{j=1}^M
    \sum_{r=0}^{n_j-1}\frac{s^r}{r!} .
\end{equation}
Equivalently, using the generating function identity
\begin{equation}
    \sum_{n\geq1} x^n\sum_{r=0}^{n-1}\frac{s^r}{r!}
    =
    \frac{x e^{sx}}{1-x} ,
\end{equation}
the tail probability can be written as

\begin{equation}
\label{eq:Srand-tail}
    \PP(S_{\rm rand}\geq s)
    =
    \sum_{k=0}^{N-M}
   \frac{\binom{N-k-1}{M-1}}{\binom{N-1}{M-1}}
    \frac{(Ms)^k}{k!}e^{-Ms}.
\end{equation}
This should be compared with the optimal tail in~\eqref{eq:claimed_identity}.
The two coincide in the two trivial cases $M=1$ and $M=N$; in general, $S^*\geq S_{\rm rand}$ and indeed, from $\binom{N-k-1}{M-1}\leq \binom{N-1}{M-1}$, we have a quantitative control on the inequality
\begin{equation}
   \PP\left(S^*\geq s\right)\geq  \PP\left(S_{\rm rand}\geq s\right).   
\end{equation}
\end{rem}

From the explicit formulae~\eqref{eq:density_S^*} and~\eqref{eq:density_Stilde}, the large-$N$ behavior follows by routine
analysis. We separate the typical  fluctuations from the
large–deviation regime. 
\subsubsection{Typical fluctuations}
Let $\alpha=M/N\in[0,1]$. As $N\to\infty$. 
\begin{enumerate}
    \item if $\alpha\in[0,1)$, then
       \begin{equation}
   \label{eq:conv_S^*1}
       \frac{M}{\sqrt{N-M+1}}\left(S^*-\frac{N-M+1}{M}\right)\stackrel{d}{\to}\mathcal{N}\left(0,1\right);
   \end{equation}
   \item  if $N-M+1=k$ is fixed (so that $\alpha=1$), then, 
   \begin{equation}
          \label{eq:conv_S^*2}
            MS^*\stackrel{d}{=}\operatorname{Gamma}(k,1).
   \end{equation}
\end{enumerate}
Similarly,
\begin{enumerate}
   \item if $\alpha\in(0,1)$, then
\begin{equation}
   MN \sqrt{\frac{N+1}{(M-1) (N-M+1)}}\left(\widetilde S-\frac{N-M+1}{MN}\right)\stackrel{d}{\to}  \mathcal{N}\left(0,1\right);
\end{equation}
\item if $M-1=k$ is fixed (so that $\alpha=0$),  then 
\begin{equation}
\label{eq:conv_Stilde_alpha0}
    N(1-M\widetilde S)
    \stackrel{d}{\to}
    \operatorname{Gamma}(k,1).
\end{equation}
\item if $N-M+1=k$ is fixed (so that $\alpha=1$), then 
\begin{equation}
\label{eq:conv_Stilde_alpha1}
    (M-1)M\widetilde S
    \stackrel{d}{\to}
    \operatorname{Gamma}(k,1) .
\end{equation}
\end{enumerate}
For a numerical illustration, see Fig.~\ref{fig:largesize}.

\subsubsection{Large deviations}
Using Stirling approximation on~\eqref{eq:density_S^*}-\eqref{eq:density_Stilde}, one gets the following large deviation formulae for the density of $S^*$ and $M\widetilde{S}$. (Formula~\eqref{eq:subleading-exp} can be alternatively obtained by specializing Theorem~\ref{thm:LD}.)

Let $0<\alpha<1$. Then, as $N\to\infty$ with $M=\lfloor \alpha N\rfloor$,
\begin{align}
\label{eq:subleading-exp}
-\frac{1}{N}\log f_{S^*}(s)
&=
\psi_{S^*}(s)
+ \frac{1}{N}\log\sqrt{\frac{2\pi (1-\alpha)}{\alpha^2 N}}
+ o\left(\frac{1}{N}\right),\\
\label{eq:subleading-exp_rel}
    -\frac{1}{N}\log f_{M \widetilde S}(s)
&=
\psi_{\widetilde{S}}(s)
+ \frac{1}{N}\log\left((1-s)^2\sqrt{\frac{2\pi(1-\alpha)}{\alpha^3 N}}\right)
+ o\!\left(\frac{1}{N}\right),
\end{align}
where,
\begin{align}
\psi_{S^*}(s) &= \alpha s - (1{-}\alpha)
  - (1{-}\alpha)\log\left(\frac{\alpha s}{1{-}\alpha}\right),\quad &&\text{for $s>0$},
  \label{eq:psi-exp}\\
  \psi_{\widetilde S}(s)&= - \alpha \log \left(\frac{1-s}{\alpha}\right) -(1-\alpha) \log \left(\frac{s}{1-\alpha}\right),\quad&&\text{ for $0<s<1$}.
  \label{eq:psi-exp_rel}
\end{align}
For a numerical illustration, see Fig.~\ref{fig:LDF}.

\begin{figure}
    \centering
\includegraphics[width=.8\linewidth]{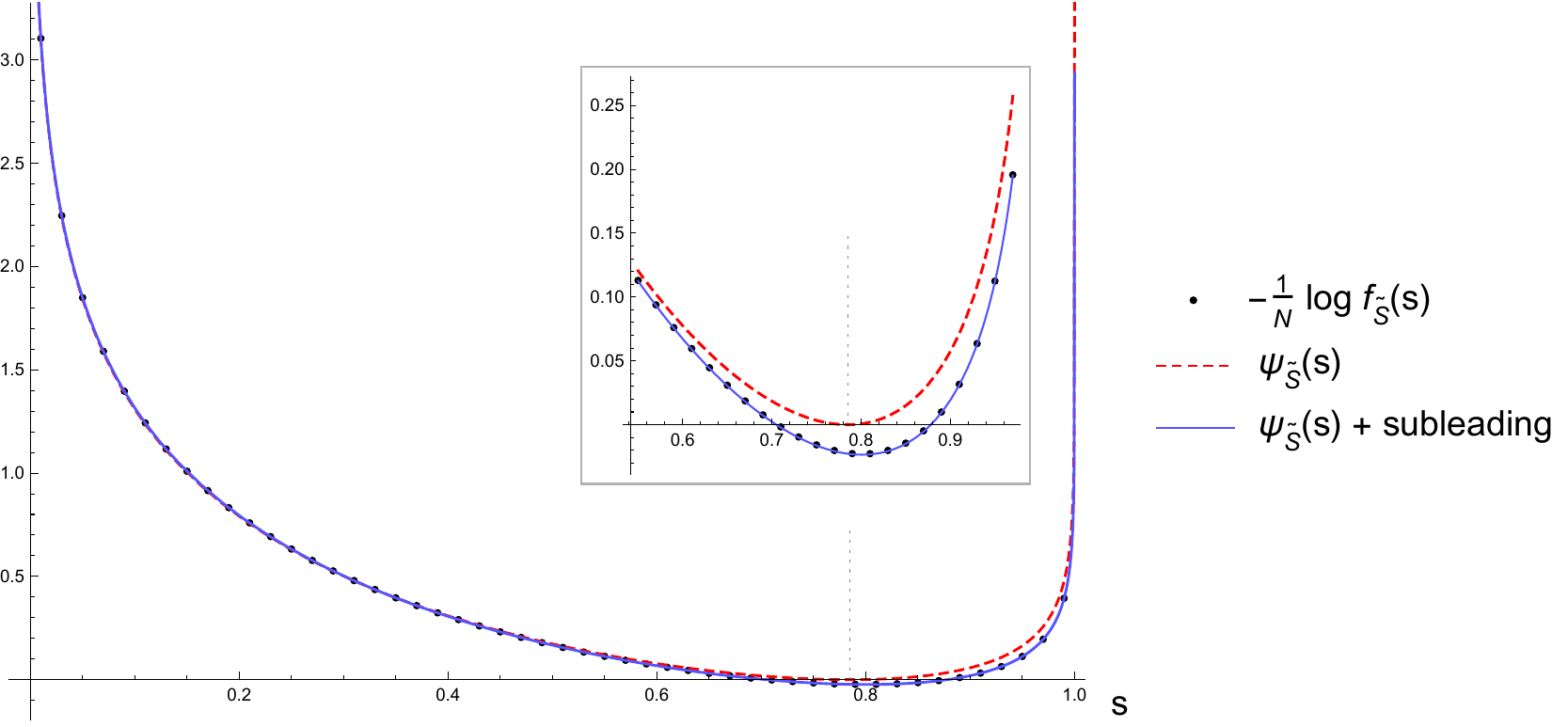}
    \caption{Comparison of $-(1/N)\log f_{M\widetilde{S}}(s)$ (black dots) for exponential gaps with $N=100$ and $\alpha=0.215$, against the rate function $\psi_{\widetilde{S}}(s)$ (red dashed curve) in~\eqref{eq:psi-exp_rel} and the asymptotics~\eqref{eq:subleading-exp_rel} including the first subleading correction (blue solid curve). The inset is a zoom around the typical value $\tilde{s}= 0.785...$. }
    \label{fig:LDF}
\end{figure}

\subsection{Geometric gaps} 
We now consider a discrete case where the random points $P_0<P_1<\cdots<P_{N}$ are located on the integers, so that the max-min spacing $S^*$ is integer-valued.

Suppose that the $T_j$'s are independent geometric random variables  with parameter $0<p<1$:
\[
  \PR(T_{{j}}=k)=(1-p)^{k-1}\,p, \qquad k=1,2,3,\ldots
\] 
We can compute  the PGF of the first-passage time $\tau_1(s)$ using~\eqref{eq:pgf-telescopic}:
\[
p_s(z) = \EE[z^{\tau_1(s)}]=z-(1-z)\sum_{k=1}^{s-1}\PP\left(P_k<s\right)z^k,
\]
where we used $\PP\left(P_k<s\right)=0$ for all $k=s,s+1,\ldots$. The distribution of $P_k$ (a sum of i.i.d. geometric random variables) is a negative binomial distribution, and this gives ($q=1-p$):
\begin{align}
  p_s(z) = \EE[z^{\tau_1(s)}]=z\,(q+pz)^{s-1},\quad \text{for all $s=1,2,\ldots$}.
  \label{eq:pgf-geom}
\end{align}
Applying formula~\eqref{eq:main}, for integer thresholds $s$, we get
\begin{equation}
\label{eq:max-min-bin}
    \PR(S^*\ge s)
  =  \sum_{k=0}^{N-M}\binom{M(s-1)}{k}\,p^k\,q^{M(s-1)-k},\quad s=1,2,\ldots.
\end{equation}
Equivalently,
\begin{equation}
\label{eq:max-min-bin2}
    \PR(S^*\ge s)
  = \PR\left(\operatorname{Bin}(M(s-1),p)\le N{-}M\right),
\end{equation}
where $\operatorname{Bin}(n,p)$ denotes a \textbf{binomial} random variable with  distribution $p(k)=\binom{n}{k}p^k(1-p)^{n-k}$, for $k=0,\ldots, n$.

Formula~\eqref{eq:max-min-bin2} is the geometric (discrete) analogue of~\eqref{eq:claimed_identity_Poisson} for exponential (continuous) spacings.

\begin{rem}
\label{rem:memoryless}
    The memoryless property of exponential and geometric gaps is reflected in the distribution of the
first-passage time:
\begin{align}
    \tau_1(s)-1 &\stackrel{d}{=} \operatorname{Poisson}(s),\quad \text{for exponential gaps},\\
\tau_1(s)-1&\stackrel{d}{=}\operatorname{Bin}(s-1,p),\quad \text{for geometric gaps}.
\end{align}
The reproductive property of the Poisson and Binomial distribution leads to the closed formulae~\eqref{eq:claimed_identity_Poisson} and ~\eqref{eq:max-min-bin2}, respectively.
\end{rem}

\section{Iterative Interval Refinement Method}
\label{sec:num}

For large values of $N$ and $M$, an exhaustive search over the
$|I_{M,N}|=\binom{N-1}{M-1}$ possible selections of points is computationally
infeasible. The correspondence to the  threshold-resetting (Lemma~\ref{lem:equiv}) suggests an iterative
procedure to compute the max--min spacing, together with an optimal selection
of $M+1$ points. The method
can be viewed as a variant of the bisection method for root finding.
{Threshold-search procedures are classical in the literature on $p$-dispersion, although the feasibility problem associated with a fixed threshold is treated differently in the general and in the one-dimensional settings. In the general setting, ~\cite{Erkut1990} introduced the fixed-threshold $r$-separation problem and proposed integer-programming and
branch-and-bound procedures within a search over the separation parameter. Other approaches also include maximum-independent-set feasibility
tests~\cite{SayahIrnich2017} or binary search over the sorted set of
pairwise distances, followed by candidate testing through set-packing feasibility problem ~\cite{FrancoUchoa2015}.

For points already ordered on a line, the feasibility test is considerably simpler: no integer program or graph-optimization problem has to be solved, since feasibility at a prescribed threshold is decided exactly by the greedy construction detailed below. This is the same basic algorithmic principle commonly presented in the competitive-programming literature under the name ``aggressive cows'' (see, for example, \cite{aggressiveCows}){, where one wishes to assign $M+1$ cows to $N+1$ stalls such that the minimum distance between any two cows is maximized}.
Other specialized exact algorithms exploiting the ordering of points on a line also include~\cite{RRT94} and~\cite{AkagiEtAl2018,araki2022max}.

The procedure used here builds on this classical one-dimensional
feasibility mechanism, but expresses it directly in terms of the
threshold-resetting construction of Lemma~\ref{lem:equiv}: for each
trial threshold, the indices selected by the greedy scan coincide with
the successive threshold-crossing indices. Combining this test with
bisection therefore provides a simple numerical implementation of the
probabilistic construction developed above. We now describe the
procedure explicitly in the notation of the present paper.
}

Let $ P_0<P_1<\cdots<P_N$
be the ordered initial point configuration, not necessarily random. 
Observe that, for all $M\leq N$, the max--min spacing $s^*$ satisfies
\begin{equation}
    0\leq s^* \leq \frac{P_N-P_0}{M}.
\end{equation}
We therefore initialise the search interval as
\begin{equation}
    [a,b]=\left[0,\frac{P_N-P_0}{M}\right].
\end{equation}
At each iteration, we bisect the current interval and set
\begin{equation}
    s=\frac{a+b}{2}.
\end{equation}
We then test whether the threshold $s$ is feasible by the greedy construction
suggested by Lemma~\ref{lem:equiv}. Starting from the left endpoint, set
\begin{equation}
    j_0:=0,
\end{equation}
and define recursively
\begin{equation}
    j_r:=\min\{\ell>j_{r-1}: P_\ell-P_{j_{r-1}}\geq s\},
    \qquad r=1,\ldots,M,
    \label{eq:greedy-indices}
\end{equation}
whenever the set on the right-hand side is non-empty. If at some step the set
is empty, the construction stops and the threshold $s$ is declared infeasible.

By Lemma~\ref{lem:equiv}, the threshold $s$
is feasible if and only if the greedy construction completes $M$ steps no
later than the final point, namely
\begin{equation}
    s^*\geq s
    \quad\Longleftrightarrow\quad
    j_M\leq N.
    \label{eq:feasibility-test}
\end{equation}
Equivalently, if \eqref{eq:greedy-indices} produces indices
$0=j_0<j_1<\cdots<j_M\leq N$, then the points
\begin{equation}
    P_{j_0},P_{j_1},\ldots,P_{j_{M-1}},P_N
\end{equation}
form an admissible selection whose consecutive spacings are all at least $s$.

Accordingly, the interval is refined as follows. If the threshold $s$ is
feasible, we continue the search in the right subinterval by setting
\begin{equation}
    a:=s.
\end{equation}
If the threshold $s$ is infeasible, we continue the search in the left
subinterval by setting
\begin{equation}
    b:=s.
\end{equation}
After a prescribed number of iterations, or once the desired precision is
reached, the midpoint of the final interval is returned as an approximation of
$s^*$.
The same procedure also returns an approximately optimal selection of points:
one runs the greedy construction \eqref{eq:greedy-indices} at the final feasible
threshold and retains the corresponding points, together with the right
endpoint $P_N$.

We used Theorem~\ref{thm:main-exp} as a benchmark for validating the algorithm.
For values of $N$ and $M$ for which exhaustive search is computationally
infeasible, the numerically computed max--min spacing can be compared with the
exact distribution given in~\eqref{eq:density_S^*}. Fig.~\ref{fig:largesize}
presents the results of numerical experiments conducted on several realisations
of random points $P_0<\cdots<P_N$ with i.i.d.\ exponential spacings.

\begin{figure}
    \centering
\includegraphics[width=0.55\linewidth]{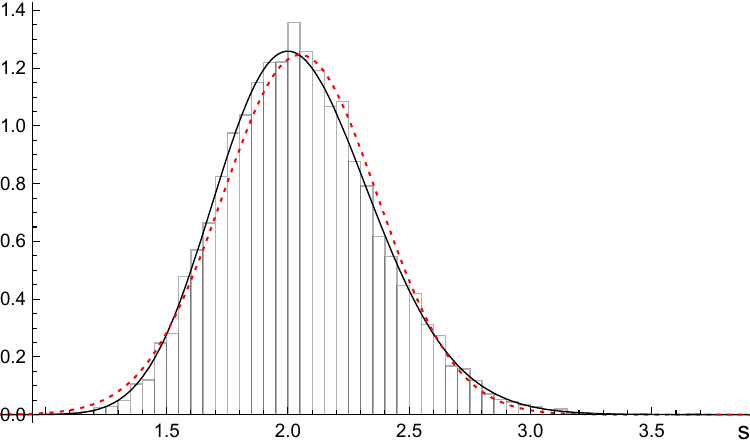}
    \caption{Max-min spacing $S^*$ for points with random i.i.d. exponential gaps. Numerical simulations (histograms) over a sample size of $10^4$  obtained using the iterative interval refined method described in Sec.~\ref{sec:num}. In black the explicit formula~\eqref{eq:density_S^*}. The dotted-red curve is the Gaussian approximation~\eqref{eq:conv_S^*1}. Here $N=60$, and $M=20$. (Total number of selections: $|I_{M,N}|=\binom{59}{19}\simeq 10^{15}$.)}
    \label{fig:largesize}
\end{figure}

\section{Proofs}
\label{sec:proofs}

%
%%%%%
\begin{proof}[Proof of Lemma~\ref{lem:equiv}] Let $(\tau_k(s))_{k\geq 0}$ be the renewal times defined in~\eqref{def:ResetTime} with $\tau_0(s):=0$. The second equivalence is immediate: $K_N(s)\geq M$ (the reset-to-zero process completes at least $M$ excursions in $N$ steps), if and only if  $\tau_M(s)\leq N$ (the reset-to-zero process reaches the threshold $s$ for the $M$-th time no later than the $N$-th step). 

Now, suppose $\tau_M(s)\leq N$.  Consider the selection of indices  defined by the first $M-1$ reset-to-zero times 
    \begin{equation}
        i^s=(i_0^s,i_1^s,\ldots,i^s_{M-1}, i_M^s)=\left(0,\tau_1(s),\ldots, \tau_{M-1}(s),N\right).
    \end{equation}
    This is a valid grouping $i^s\in I_{M,N}$: the blocks are consecutive, they cover
all $N$ gaps, and there are exactly $M$ of them.  
    By definition of the $\tau_k(s)$'s,
    \begin{equation*}
\begin{aligned}
 P_{i^s_1}-P_{i^s_0}=   P_{\tau_1(s)}-P_{0}&\geq s,\\ 
  P_{i^s_2}-P_{i^s_1}=  P_{\tau_2(s)}-P_{\tau_1(s)}&\geq s,\\
   \vdots &\\
    P_{i^s_{M-1}}-P_{i^s_{M-2}}=    P_{\tau_{M-1}(s)}-P_{\tau_{M-2}(s)}&\geq s\\
   P_{i^s_M}-P_{i^s_{M-1}}=  P_{N}-P_{\tau_{M-1}(s)}&\geq P_{\tau_M(s)}-P_{\tau_{M-1}(s)}\geq s. 
\end{aligned}    
    \end{equation*}
All spacings arising by this selection of indices in $I_{M,N}$ are at least $s$.  Therefore, 

$$
S^*\geq \min_{1\leq j\leq M}(P_{i^s_{j}}-P_{i^s_{j-1}})\geq s.
$$
Conversely, assume that $S^*\geq s$. Then, there exists a selection of indices
$$0=i^*_0< i^*_1<\cdots<i^*_{M-1}<i^*_M=N$$
such that for all $1\leq j\leq M$,
$$
P_{i^*_{j}}-P_{i^*_{j-1}}\geq s.
$$
Since the $\tau_k$'s are first-passage times, 
    \begin{equation*}
\begin{aligned}
\tau_1(s)=i^s_1&\leq i^*_1,\\ 
 \tau_2(s)=i^s_2&\leq i^*_2,\\
   \vdots &\\
   \tau_{M-1}(s)=i^s_{M-1}&\leq i^*_{M-1}\\
   \tau_M(s)&\leq i^*_M=N. 
\end{aligned}    
    \end{equation*}
Hence $\tau_M(s)\leq N$.
\end{proof}

\begin{proof}[Proof of Theorem~\ref{thm:main}] 
As already remarked, the sequence $(\tau_k(s))_{k\geq 0}$ forms a renewal process: the cycle lengths $L_1:=\tau_1(s)-\tau_0(s) $, $L_2:=\tau_2(s)-\tau_1(s)$,\ldots, are i.i.d. copies of the first-passage time $\tau_1(s)$. 

By Lemma~\ref{lem:equiv}, $S^*\geq s$ if and only if $\tau_M(s)=L_1+\cdots+L_M\leq N$. 
Therefore,
\begin{equation}
\PP\left(S^*\geq s\right)=\PP\left(\tau_M(s)\leq N\right)=\sum_{n=M}^N \PP\left(\tau_M(s)=n\right).
\end{equation}
Consider the PGF of $\tau_M(s)$:
\begin{equation}
\EE[z^{\tau_M(s)}]=\sum_{n\geq1}\PP\left(\tau_M(s)=n\right)z^n. 
\end{equation}
Then,
\begin{equation}
   \PP\left(S^*\geq s\right)=\sum_{n=M}^N[z^n] \EE[z^{\tau_M(s)}].
\end{equation}
Since $L_1,\ldots,L_M$ are independent, and distributed as $\tau_1(s)$,
\begin{equation}
    \EE[z^{\tau_M(s)}]=   \EE[z^{L_1+\cdots+L_M}]=\EE[z^{L_1}]\cdots\EE[z^{L_M}]=\EE[z^{\tau_1(s)}]^M=p_s(z)^M.
\end{equation} 
We conclude that
\begin{equation}
   \PP\left(S^*\geq s\right)=\sum_{n=M}^N[z^n] p_s(z)^M=\sum_{n=0}^N[z^n] p_s(z)^M=[z^N]\frac{p_s(z)^M}{1-z}.
\end{equation}
In the last steps we used the fact that the power series $p_s(z)^M$ has no term lower than $z^M$, and we then multiplied by a geometric series to obtain the cumulative sum of the first $N$ coefficients of the power series.
\end{proof}

\begin{proof}[Proof of Lemma~\ref{lem:telescopic}]
    The distribution of the first-passage time $\tau_1(s)$ is determined by the
law of gaps $T_j$'s. Indeed,
\begin{equation}
    \tau_1(s)
    =\min\left\{k>0: P_k\geq s\right\}
    =\min\left\{k>0: T_1+\cdots+T_k\geq s\right\}.
\end{equation}
For $k\geq1$, 
\begin{equation}
   \tau_1(s)=k \iff P_{k-1}<s \, \text{ and }\,   P_k\geq s,
\end{equation}
with the convention $P_0=0$. Since the increments $T_j$'s are strictly positive, the sequence $(P_k)_{k\geq0}$ is increasing, and therefore
$P_k<s \Rightarrow P_{k-1}<s$. Hence
\begin{align}
  \PR(\tau_1(s)=k)
  &= \PR(P_{k-1}<s,\ P_k\ge s) \notag\\
  &= \PR(P_{k-1}<s)-\PR(P_{k-1}<s,\ P_k<s) \notag\\
  &= \PR(P_{k-1}<s)-\PR(P_k<s).
  \label{eq:PL_proof}
\end{align}
Thus the distribution of $\tau_1(s)$ can be written as a
telescopic difference of the distribution functions of the partial sums
$P_{k-1}$ and $P_k$:
\begin{equation}
\begin{aligned}
    p_s(z)
    &=\sum_{k\geq1}
    \left[\PR(P_{k-1}<s)-\PR(P_k<s)\right]z^k\\
      &=\sum_{k\geq1}\PR(P_{k-1}< s)z^{k}-\sum_{k\geq1}\PR(P_{k}< s)z^{k}\\
    &=\sum_{k\geq0}\PR(P_{k}< s)z^{k+1}-\sum_{k\geq1}\PR(P_{k}< s)z^{k}\\
      &=z-(1-z)\sum_{k\geq1}\PR(P_{k}< s)z^{k},
    \label{eq:pgf-telescopic_proof}
    \end{aligned}
\end{equation}
where in the last line we used $\PP(P_{0}< s)=1$ for all $s>0$.

We now prove that, if the increments $T_j$'s are strictly positive, then $p_s(z)$ is entire in $z$. From~\eqref{eq:pgf-telescopic_proof}, it is enough to show that the series 
\begin{equation}
\label{eq:series_P_k}
   \sum_{k\geq1}\PR(P_{k}< s)z^{k} 
\end{equation}
is convergent for all $z\in\C$.
By the Cauchy--Hadamard theorem, the radius of convergence of~\eqref{eq:series_P_k} is
% \[
% p_s(z)=\sum_{k\ge 1}\mathbb P(\tau_1(s)=k)z^k
% \]
\begin{equation}
\label{eq:Cauchy-Hadamard}
  R=\frac{1}{\displaystyle
\limsup_{k\to\infty}\mathbb P(P_{k}< s)^{1/k}}.  
\end{equation}
We shall prove that, for all $s>0$,
\begin{equation}
\label{eq:claimed_limsup}
\limsup_{k\to\infty}\mathbb P(P_{k}< s)^{1/k}=0.
\end{equation}
Let $0<a<s$, and write
\begin{equation}
q(a)=\mathbb P(T_1\le a).
\end{equation}
Since $\mathbb P(T_1=0)=0$,  $q(a)\downarrow0$, as $a\downarrow0$.
Let 
$$ m=\left\lfloor \frac{s}{a}\right\rfloor.
$$
If $P_k< s$, then at most $m$ of the $k$ jumps $T_i$ can exceed $a$. Indeed, if
$m+1$ jumps were larger than $a$, then $
P_k> (m+1)a>s$.
 Therefore, for all $0\leq m\leq k$,
% \[
% \{P_n< s\}
% \subseteq
% \left\{\#\{1\le i\le n:T_i\geq a\}\le m\right\}.
% \]
% Thus
\[
\mathbb P(P_k\le s)
\le
\mathbb P\left(\#\{1\le i\le k:T_i> a\}\le m\right).
\]
Since the variables are independent, the number of jumps exceeding $a$ is a
binomial random variable with parameters $k$ and $1-q(a)$. Hence
\begin{equation}
\mathbb P(P_k\le s)
\le
\sum_{j=0}^{m}
\binom{k}{j}
(1-q(a))^j q(a)^{k-j}.
\end{equation}
We  estimate this upper bound as 
\[
\mathbb P(P_k\le s)
\le
q(a)^{k-m} 
\sum_{j=0}^{m}\binom{k}{j}\leq q(a)^{k-m} 2^k .
\]
Consequently,
\[
\limsup_{k\to\infty}\mathbb P(P_k < s)^{1/k}\leq \limsup_{k\to\infty}\mathbb P(P_k\le s)^{1/k}
\le 2 q(a),
\]
for all $a>0$, and hence the claim~\eqref{eq:claimed_limsup}.
\end{proof}
The following proof of Theorem~\ref{thm:LD}, is based on the steepest descent method. We will need to deform  the contour of integration in Cauchy's integral formula. These deformations are justified since $p_s(z)$ is analytic in the whole complex plane.

\begin{proof}[{Proof of Theorem~\ref{thm:LD}}]

{By Theorem~\ref{thm:main}} and Cauchy integral formula,
\begin{equation}\label{eq:P}
\PR(S^*\ge s)
  = [z^N] \frac{p_s(z)^M}{(1-z)}=\frac{1}{2\pi \mathrm{i}}\oint_{\gamma}
  \frac{p_s(z)^M}{(1-z)\,z^{N+1}}\,dz= \frac{1}{2\pi \mathrm{i}}\oint_{\gamma}
  \frac{e^{Ng(z)}}{(1{-}z)\,z}\,dz ,
\end{equation}
where $\gamma$ is a positively oriented simple contour  entirely contained in the unit disk and enclosing $z=0$, and 
\begin{equation}
    g(z)=\alpha\log p_s(z)-\log z.
\end{equation} 
By Lemma~\ref{lem:telescopic}, the integrand in~\eqref{eq:P} is meromorphic with poles at $z=0$ and $z=1$, only. 
Choose $\gamma$ to be a circle $C_r$ of radius $r$ centred at $0$. 
Since  $p_s(z)$ is a power series with positive coefficients, 
\begin{equation}
    |p_s(r e^{\mathrm{i} \theta})|=\left|\sum_{k\geq1} \mathbb{P}(\tau_1(s)=k){r^k e^{\mathrm{i}k\theta}}\right|\leq\sum_{k\geq1} \mathbb{P}(\tau_1(s)=k){r^k} =p_s(r),
\end{equation}
therefore
$\operatorname{Re}g(r{e^{\mathrm{i}\theta}})
$ has a global maximum at $\theta=0$.
For all $s>0$, choose $r=z(s)$, where $z(s)$ is a real and positive solution of $g'(z)=0$, i.e. $ \mu(z,s)=1/\alpha$, where $$\mu(z,s)=z \frac{p'_s(z)}{p_s(z)}=\EE [\tau_1(s,z)] .$$

The meaning of this saddle-point equation is the following.  The event
$\{S^*\geq s\}$ is equivalent, by Lemma~\ref{lem:equiv}, to completing $M$
renewal cycles by time $N$.  Hence, if $M/N=\alpha$, the
typical length of one cycle is $N/M=1/\alpha$.  The tilted law
$\tau_1(s,z)$ is precisely the exponential change of measure under which the
first-passage time has mean
\begin{equation}
    \mu(z,s)=\EE[\tau_1(s,z)] .
\end{equation}
Thus the saddle-point condition
\begin{equation}
    \mu(z,s)=\frac{1}{\alpha}
\end{equation}
selects the value of the tilt for which a typical renewal cycle has the length
required by the constraint.  In particular, $z<1$ biases the renewal process
towards shorter cycles, while $z>1$ biases it towards longer cycles.

By assumption, for each $s>0$ there exists a positive
solution $z(s)>0$ of the saddle-point equation
\begin{equation}
    \mu(z,s)=\frac{1}{\alpha},
    \qquad
    \mu(z,s):=
    z\frac{p_s'(z)}{p_s(z)}
    =
    \EE[\tau_1(s,z)] .
\end{equation}
Moreover, from \eqref{eq:TiltedMeanVar},
\begin{equation}
    \frac{d}{dz}\mu(z,s)
    =
    \frac{1}{z}\operatorname{Var}[\tau_1(s,z)].
\end{equation}
Hence, whenever $\operatorname{Var}(\tau_1(s,z))>0$, such a solution is unique.

Computing the second derivative, we have \begin{equation}
\label{eq:var_tauz}
    g''(z(s))=\alpha \frac{\operatorname{Var}[\tau_1(s,z(s))]}{z(s)^2}>0 .
\end{equation} 
 
On the circle $z=z(s)e^{\mathrm{i}\theta}$, the saddle equation gives
\begin{equation}
    g(z(s)e^{\mathrm{i}\theta})
    =
    g(z(s))
    -
    \frac{\alpha\sigma^2(s)}{2}\theta^2
    +
    O(\theta^3),
    \qquad \theta\to0 ,
\end{equation}
where $\sigma^2(s)=\operatorname{Var}(\tau_1(s,z(s)))$.
Thus $\theta=0$ is a nondegenerate maximum of
$\operatorname{Re}g(z(s)e^{\mathrm{i}\theta})$ along the circular contour.

Let
\begin{equation}
    \label{eq:rate_psi_g}
\psi(s)=-g(z(s))=-\operatorname{Re}g(z(s)).
\end{equation}
(Since $z(s)>0$, the quantity $g(z(s))$ is real.) 
Finally, the steepest descent direction for $\operatorname{Re}g(z)$ is parallel to the imaginary axis, and it is therefore aligned with the tangent to  $C_{z(s)}$ at $z(s)$. We are now prepared to apply the steepest descent method. 
\begin{figure}
\includegraphics[width=\textwidth]{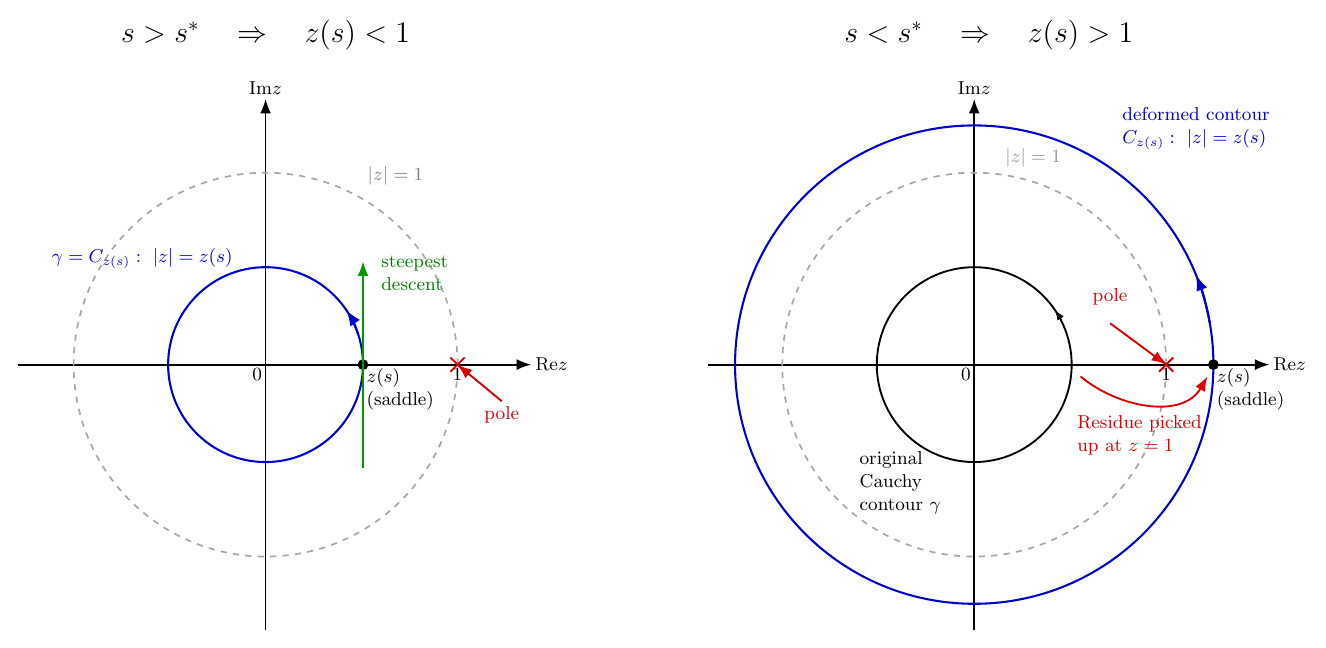}
\caption{Contours of integration appearing in the calculation of $\PR(S^*\geq s)$}
\label{fig:pic}
\end{figure}
By definition, $s^*$ is the solution of $\mu(1,s^*)=1/\alpha$. Hence $z(s^*)=1$ and, by monotonicity, %, $z(s)<1$ for $s>s^*$, while $z(s)>1$ for $s<s^*$. 
    \begin{equation}\label{eq:wts}
    \begin{cases}
        &s<s^* \Rightarrow z(s)>1\\
        &s=s^* \Rightarrow z(s)=1\\
        &s>s^*\Rightarrow z(s)<1 .
    \end{cases}
\end{equation}
Therefore, if $s>s^*$, the contour $C_{z(s)}$ is contained in the unit disk where the integrand $\frac{e^{Ng(z)}}{(1-z)z}$ has a single pole at $z=0$ (see Figure~\ref{fig:pic}). Therefore,
\begin{equation}
\begin{aligned}
\PP(S^*\geq s)
&=
\frac{1}{2\pi}
\int_{-\pi}^{\pi}
\frac{
\exp\{N g(z(s)e^{\mathrm{i}\theta})\}
}
{1-z(s)e^{\mathrm{i}\theta}}
\,d\theta\\
&=
\frac{e^{-N\psi(s)}}{
(1-z(s))\sqrt{2\pi\alpha N\sigma^2(s)}
}
\left[1+o(1)\right],
\end{aligned}
\end{equation}
which yields~\eqref{eq:right-box}.

For $s<s^*$, the contour deformation from $\gamma$ to $C_{z(s)}$ crosses the pole at $z=1$. Therefore,
\begin{equation}
 \PR(S^*\ge s)
 =\frac{1}{2\pi \mathrm{i}}\oint_{\gamma}
  \frac{e^{Ng(z)}}{(1{-}z)\,z}\,dz= \frac{1}{2\pi \mathrm{i}}\oint_{C_{z(s)}}
  \frac{e^{Ng(z)}}{(1{-}z)\,z}\,dz-\operatorname{Res}_{z=1}\left(\frac{e^{Ng(z)}}{(1{-}z)\,z}\right) . 
\end{equation}
 The residue is:
 \begin{equation}
\operatorname{Res}_{z=1}\left(\frac{e^{Ng(z)}}{(1{-}z)\,z}\right)=\lim_{z\to 1}(z-1)\frac{p_s(z)^M}{(1-z)z^{N+1}}=-p_s(1)^M=-1.
 \end{equation}
 Therefore,
\begin{equation}
\begin{aligned}
\PP(S^*<s)=1- \PR(S^*\geq s)
&=
-\frac{1}{2\pi}
\int_{-\pi}^{\pi}
\frac{
\exp\{N g(z(s)e^{\mathrm{i}\theta})\}
}
{1-z(s)e^{\mathrm{i}\theta}}
\,d\theta  \\
&=
\frac{e^{-N\psi(s)}}{
(z(s)-1)\sqrt{2\pi\alpha N\sigma^2(s)}
}
\left[1+o(1)\right],
\end{aligned}
\end{equation}
which matches~\eqref{eq:left-box}.
\end{proof}

\begin{proof}[Proof of Theorem~\ref{thm:main-exp}]
We prove the identity in distribution~\eqref{eq:identity_S^*} by computing the tail. 
It suffices to show~\eqref{eq:claimed_identity}  for all $s\ge0$.
(Differentiation then yields the density \eqref{eq:density_S^*}.) 
The identity is a specialisation of Theorem~\ref{thm:main}. Recall that $P_k=T_1+\cdots+T_{k}\stackrel{d}{=}\operatorname{Gamma}(k,1)$, for all $k\geq1$. Then, by~\eqref{eq:PL}, 
\begin{equation}
\begin{aligned}
\PP(\tau_1(s)=k)&=
\PP(P_{k-1}< s)- \PP(P_{k}< s)\\
&=\PP(P_k\geq s)- \PP(P_{k-1}\geq s)\\
&=
e^{-s}\sum_{j=0}^{k-1}\frac{s^j}{j!}-e^{-s}\sum_{j=0}^{k-2}\frac{s^j}{j!}=
e^{-s}\frac{s^{k-1}}{(k-1)!}.
\end{aligned}
\end{equation}
The PGF is therefore
\begin{equation}
    p_{s}(z)=\sum_{k\geq1}\PP(\tau_1(s)=k)\;z^k= z\,e^{s(z-1)}.
\end{equation}
The coefficient of $z^n$ of $p_{s}(z)^M$ for $n\ge M$ is
\[
  [z^n]\,z^M\,e^{Ms(z-1)}
  = [z^{n-M}]\,e^{Ms(z-1)}
  = e^{-Ms}\cdot\frac{(Ms)^{n-M}}{(n{-}M)!}.
\]
Therefore, by Theorem~\ref{thm:main}:
\begin{equation}\label{eq:poisson}
  \PR(S^*\geq s) = [z^N]\frac{p_{s}(z)^M}{1-z}=\sum_{n=M}^N[z^n]p_{s}(z)^M=e^{-Ms}\sum_{k=0}^{N-M}\frac{(Ms)^k}{k!}
\end{equation}
as claimed. 

\par

We now prove the  identity~\eqref{eq:identity_Stilde}. 
Let
\begin{equation}
    W:=T_1+\cdots+T_N=P_N-P_0
\end{equation}
and introduce the normalised points
\begin{equation}
    Y_j:=\frac{P_j}{W}
    =\frac{T_1+\cdots+T_j}{T_1+\cdots+T_N},
    \qquad \text{for $j=1,\ldots,N-1$},
\end{equation}
and $Y_0:=0$ and $Y_N:=1$. 
A classical calculation, shows that $W$ and $(Y_1,\ldots,Y_{N-1})$ are independent, with
\begin{equation}
    W\stackrel{d}{=}\operatorname{Gamma}(N,1),\quad     (Y_1,\ldots,Y_{N-1})
    \stackrel{d}{=}
    (U_{1:N-1},\ldots,U_{N-1:N-1}),
\end{equation}
where $U_{1:N-1}<\cdots<U_{N-1:N-1}$ are the order statistics of
$N-1$ independent uniform random variables on $[0,1]$.

Indeed, the change of variables
\begin{equation}
\begin{cases}
    t_1 = w y_1,\\
    t_j = w(y_j-y_{j-1}), \qquad j=2,\ldots,N-1,\\
    t_N = w(1-y_{N-1}),
\end{cases}
\end{equation}
maps the region $t_1,\ldots,t_N>0$ onto
\begin{equation}
    w>0,\qquad 0<y_1<\cdots<y_{N-1}<1.
\end{equation}
Its Jacobian is $w^{N-1}$ (see, e.g.~\cite[Appendix B.1]{Cunden20}). Since the $T_j$'s are i.i.d.\
$\operatorname{Exp}(1)$, their joint density is $e^{-(t_1+\cdots+t_N)}
    =e^{-w}$.    
Therefore,
\begin{equation}
\begin{aligned}
f_{Y_1,\ldots,Y_{N-1},W}(y_1,\ldots,y_{N-1},w) &=
w^{N-1}e^{-w}
\bm 1_{\{w>0\}}
\bm 1_{\{0<y_1<\cdots<y_{N-1}<1\}}  \\
&\qquad =
\frac{w^{N-1}e^{-w}}{\Gamma(N)}
\bm 1_{\{w>0\}}
\cdot
(N-1)!\,
\bm 1_{\{0<y_1<\cdots<y_{N-1}<1\}}.
\end{aligned}
\end{equation}
This proves the claim.

We now use the homogeneity of the max--min spacing: multiplying all
points by a positive constant multiplies all spacings, and hence also the
max--min spacing, by the same constant. Hence,
\begin{equation}
    S^*
    =
    W\,\widetilde S,
\end{equation}
where $\widetilde S$ is the max--min spacing of the normalised configuration
\begin{equation}
    0=Y_0<Y_1<\cdots<Y_{N-1}<Y_N=1.
\end{equation}
By the independence just proved, $\widetilde S$ is independent of $W$, and
the normalised configuration has the same law as
\begin{equation}
    0=U_{0:N-1}<U_{1:N-1}<\cdots<U_{N-1:N-1}<U_{N:N-1}=1.
\end{equation}
Thus the relative max--min spacing for exponential gaps coincides in law with
the max--min spacing of $N-1$ independent uniform points in the unit
interval, together with the two endpoints.

It remains to identify its density. Since $S^*=W\widetilde S$, with
$W\sim\operatorname{Gamma}(N,1)$ independent of $\widetilde S$, the density
of $S^*$ is the Mellin convolution
\begin{equation}
    f_{S^*}(s)
    =
    \int_0^\infty
    f_{\widetilde S}\!\left(\frac{s}{w}\right)
    f_W(w)\,\frac{dw}{w}.
\end{equation}
Substituting the claimed density~\eqref{eq:density_Stilde}
and
\begin{equation}
    f_W(w)=\frac{w^{N-1}e^{-w}}{\Gamma(N)},\qquad w>0,
\end{equation}
we get
\begin{equation}
\begin{aligned}
f_{S^*}(s)=
\binom{N-1}{M-1}M(M-1)(Ms)^{N-M}
\frac{1}{\Gamma(N)}
\int_{Ms}^{\infty}
\left(1-\frac{Ms}{w}\right)^{M-2}
w^{M-1}e^{-w}\,\frac{dw}{w}.
\end{aligned}
\end{equation}
By the change of
variables $r=w-Ms$, the integral evaluates to
$\Gamma(M-1)e^{-Ms}$, and hence we get exactly the density in~\eqref{eq:density_S^*}. This proves the claimed
density~\eqref{eq:density_Stilde} of $\widetilde S$, and therefore the distributional identity~\eqref{eq:identity_Stilde}.

\end{proof}

\section{Conclusions and Outlook} \label{sec:conclusions}

We considered the following selection problem on the line. Starting from
$N+1$  points with i.i.d.\ positive spacings, retain $M+1$ of them so
that the minimal spacing between consecutive retained points is as large as
possible. Although the optimization ranges over
$\binom{N-1}{M-1}$ strongly correlated admissible selections, the resulting
random variable admits a simple renewal description.

The key observation is the equivalence, stated in Lemma~\ref{lem:equiv},
between the event $\{S^*\geq s\}$ and the completion of at least $M$
threshold-crossing cycles by a random walk with positive increments, reset to
the origin whenever it crosses the level $s$. This mapping reduces the original
combinatorial optimization problem to the study of a
first-passage generating function. The resulting distribution-free formula,
given in Theorem~\ref{thm:main}, provides the tail distribution of $S^*$ for arbitrary i.i.d.\ spacings. As a byproduct, it gives the probability that a threshold-resetting random walk with arbitrary jump distribution completes at least $M$ resetting cycles in $N$ steps.

This representation is particularly effective to analyze the asymptotic regime 
$M,N\to\infty$ with $M/N=\alpha$. The saddle-point equation
\eqref{eq:saddle} acquires a probabilistic meaning: the saddle selects the
exponentially tilted law for which the typical cycle length matches the
constraint imposed by the ratio $M/N$.  This leads to the large-deviation estimates of
Theorem~\ref{thm:LD}, with upper and lower tail asymptotics given in
\eqref{eq:right-box} and \eqref{eq:left-box}. The explicitly solvable cases
illustrate the general theory. For exponential spacings, the max--min spacing
has the Gamma distribution in \eqref{eq:identity_S^*}, while the relative max--min
spacing has the Beta distribution in \eqref{eq:identity_Stilde}. For geometric spacings,
the corresponding discrete formula is expressed in terms of binomial tails, as
shown in \eqref{eq:max-min-bin} and \eqref{eq:max-min-bin2}.

Several directions are worth exploring. The renewal structure used
throughout the paper relies on the independence of the spacings. Correlated
increments, conditioned point processes, or interacting particle systems would
require different ideas and may produce  a different limiting behavior.

A second direction concerns other scaling regimes. In this paper we mainly
focused on the proportional regime $M/N=\alpha$, together with some boundary
regimes in the exponential case. The cases in which $M$ is fixed, $N-M$ is
fixed, or $M$ grows sublinearly with $N$ should exhibit different asymptotic
scales and may connect more directly with classical extreme-spacing theory.
Similarly, a refined analysis of the critical window around the typical
value defined by \eqref{eq:typical} would clarify the crossover between typical
fluctuations and the large-deviation tails described here.

From the computational point of view, the renewal characterization suggests
efficient decision procedures for a prescribed threshold $s$, and hence
bisection-type algorithms for approximating $S^*$ in large deterministic
instances. The exact formulae obtained here, especially in the exponential
case, provide benchmarks for such algorithms far beyond the range where
exhaustive search over all admissible selections is possible. 

Finally, the {methods used to analyze this model heavily rely on its one-dimensional geometry.} Extending the present exact
approach to higher-dimensional random point configurations is a challenging open problem. Understanding whether an
analogue of the threshold-crossing representation survives in special
geometries, or whether different probabilistic structures replace it, would
help connect the exact solvability found here with broader questions in random
combinatorial optimisation and spatial selection problems.

\section*{Conflicts of Interest}

The authors have no conflicts of interest to declare.

\section*{Data Availability}
{All synthetic data presented in the paper were generated through numerical simulations and can be reproduced using the methods described herein. The corresponding simulation codes are available from the authors upon reasonable request.}

\section*{Acknowledgments}
F.D.C. and G.G. are supported by Gruppo Nazionale di Fisica Matematica GNFM-INdAM and by Istituto Nazionale di Fisica Nucleare INFN through the project QUANTUM. F.D.C.  acknowledges the support from  PRIN 2022 project 2022TEB52W-PE1-
`The charm of integrability: from nonlinear waves to random matrices', and from PNRR MUR project CN00000013 `Italian National Centre on HPC, Big Data and Quantum Computing'. G.G. acknowledges support from PNRR MUR project PE0000023-NQSTI and from the project ``Patto territoriale sistema universitario pugliese". 
P.V. acknowledges support from UKRI FLF Scheme (No. MR/X023028/1) and gratefully acknowledges a discussion with Satya N. Majumdar on threshold-resetting problems. {The authors thank the two anonymous referees for their careful reading of the manuscript, helpful comments, and for pointing out several relevant references.}

\printbibliography
\end{document}